\documentclass[runningheads,a4paper]{llncs2e/llncs}
\pdfoutput=1  

\usepackage[ruled,linesnumbered]{algorithm2e}
\usepackage{amssymb}
\usepackage{amsmath}
\usepackage{multirow}
\usepackage{graphicx}
\usepackage{caption}
\usepackage{soul}
\usepackage{subcaption}
\usepackage[colorinlistoftodos]{todonotes} 
\usepackage{url}
\urldef{\mailsa}\path|richard.connor@strath.ac.uk|
\urldef{\mailsb}\path|{lucia.vadicamo,fausto.rabitti}@isti.cnr.it|
\newcommand{\keywords}[1]{\par\addvspace\baselineskip
	\noindent\keywordname\enspace\ignorespaces#1}

\newcommand{\R}[1]{\mathbb{R}^{#1}}

\newcommand{\ns}{$n$-simplex\xspace}

\graphicspath{{figs/},{figs/tmp/},{figs/MIRFlickr/}}
\newcommand{\h}[1]{{#1}} 

\begin{document}

\title{High-Dimensional Simplexes for \h{Super}metric Search}

\titlerunning{High-Dimensional Simplexes for Metric Search}

%
%
\author{Richard Connor\inst{1} \and {Lucia Vadicamo} \inst{2} \and Fausto Rabitti\inst{2}}
\authorrunning{Connor et al.}


\institute{Department of Computer and Information Sciences,\\University of Strathclyde, Glasgow, G1 1XH, United Kingdom\\
	\and
	ISTI -
	CNR, Via Moruzzi 1, 56124 Pisa, Italy\\
	\mailsa\\
	\mailsb}
\maketitle

\begin{abstract}
In 1953, Blumenthal showed that every semi-metric space that is isometrically embeddable in a Hilbert space has the $n$-point property; we have previously called such spaces \emph{supermetric} spaces.
Although this  is  a strictly   stronger property than triangle inequality, it is nonetheless closely related and  many useful metric spaces possess it. These  include Euclidean, Cosine and Jensen-Shannon  spaces of any dimension.

A simple corollary of the $n$-point property is that, for any $(n+1)$ objects sampled from the space, there exists an $n$-dimensional simplex in Euclidean space whose edge lengths correspond to the distances among the objects. 
We show how the construction of such simplexes in  higher dimensions can be used to give arbitrarily tight lower and upper bounds on  distances within the original space.

This allows the construction of an $n$-dimensional Euclidean space, from which lower and upper bounds of the original space can be calculated, and which is itself an indexable space with the $n$-point property. 
For similarity search, the engineering tradeoffs are good: we show  significant reductions in data size and metric cost with little loss of accuracy, leading to a significant  overall improvement in  search performance.
\end{abstract}

\section{Introduction}
\subsubsection{Context}
To set the context, we are interested in searching a (large) finite  set of objects $S$ which is a subset of an infinite set $U$, where $(U, d)$ is a  metric space. The general requirement is to efficiently find members of $S$ which are similar to an arbitrary member of $U$, where the distance function $d$ gives  the only way by which any two objects may be compared. There are many important practical examples captured by this mathematical framework, see for example \cite{Chavez05,zezula2006similarity}. Such spaces are typically searched with reference to a query object $q \in U$.
 A threshold search  for some threshold $t$, based on a query $q \in U$, has the solution set  $\{s \in S \,\, \text{such that}\,\, d(q,s) \le t\}$.
 
 There are three main problems with achieving efficiency when the search space is very large. 
  Most obviously, for very large collections we always require scalability. This is achieved within metric search domains by techniques which avoid searching parts of the collection, typically by using data structures which take advantage of mathematical properties of the distance metrics used.

Secondly, distance metrics are often expensive. When the search space is large, semantic accuracy is important to avoid huge numbers of false positive results -- in the terminology of information retrieval, \emph{precision} becomes relatively more important that \emph{recall}. 
In such cases higher specificity will normally result in a much more expensive metric, for example Jensen-Shannon or Quadratic Form distances, which are much more expensive to compute.
 
Finally, the data objects themselves may be large. For example in the domain of near-duplicate image search, GIST representations give a better semantic comparison then MPEG-7, but occupy around 2KB per image \cite{connor2016quantifying}.

 Even although huge memory is nowadays available, a  large collection of large objects will still require to be paged. For example  a 32-bit architecture can typically address less than 2GB;  a collection of only one million GIST descriptors cannot be accommodated.

\subsubsection{Approaches} 
 In high-dimensional Euclidean spaces, the last two problems can be addressed by various \emph{dimensionality reduction} techniques. In outline, these techniques reduce an $n$-dimensional Euclidean space 
 space into an $m$-dimensional one, where $m < n$. This reduces both the size required to store the data, and the cost of the \h{Euclidean ($\ell_2$)} metric. However this may results  in a loss of precision, which can defeat the purpose if there is an accompanying loss of semantic accuracy with respect to the original data.
 
 In non-Euclidean metric spaces, such techniques are not applicable. There are however various other techniques which use reduced-size object surrogates for an initial indexing or filtering phase. Such techniques may  give approximate results or, if they are guaranteed to return a superset of the solution set, exact search can be performed by re-checking their output against the original data.
 
 \subsubsection{Outline of our Contribution} 
 Here, we present a new technique which can be used in either of these approaches. Using properties of finite isometric embedding, we show a mechanism which allows  spaces with certain properties to be translated into a second, smaller, space. For a metric space $(U,d)$, we describe a family of functions 
 $\phi_n$ 
 which can be created by measuring the distances among $n$ objects sampled from the original space, and which can then be used to create a \emph{surrogate} space:
 %
 \[\phi_n : (U,d) \rightarrow (\mathbb{R}^n,\ell_2)\]
%
with the property
%
\[
\ell_2(\phi_n(u_1),\phi_n(u_2)) \le d(u_1,u_2) \le g(\phi_n(u_1),\phi_n(u_2))
\]

for an associated function $g$. Further, the cost of evaluating $g$ and $\ell_2$ together is almost exactly the same as the cost of  $\ell_2$.

\h{This family of functions can be defined for any metric space which is isometrically embeddable in a Hilbert Space, or equivalently for any space that meets the $n$-point property} \cite{Connor2016:HilbertExclusion}.
The advantages of the \h{proposed} technique are that (a) the $\ell_2$ metric is very much cheaper than  some Hilbert-embeddable metrics; (b) the size of elements of $\mathbb{R}^n$ may be much smaller than elements of $U$, and (c) in many cases we can achieve both of these along with an \emph{increase} in the scalability of the resulting search space.

While not applicable to all purposes, we show that this mechanism may be used to great effect in a number of ``real-world" search spaces. Among other results, we show a benchmark best-performance for SISAP \emph{colors}~\cite{SISAP_man} data set.

\section{Related Work}
\subsubsection{Finite Isometric Embeddings}
are excellently summarised by Blumenthal \cite{blumenthal1933note}.
He uses the phrase  \emph{four-point property} to mean a space that is  4-embeddable in $3$-dimensional Euclidean space ($\ell_2^3$)\h{, i.e. if for any four points $x_1,x_2, x_3, x_4 \in U$ exist a mapping function $f:U \to \ell_2^2$ such that $\ell_2(f(x_i), f(x_j))=d(x_i,x_j)$, for $i,j=1,2,3,4$.} 
Wilson \cite{wilson1932relation} shows various properties of such spaces, and Blumenthal points out that results given by Wilson, when combined with work by Menger \cite{Menger1928}, generalise to show that some
spaces with the four-point property also have the $n$-point property: any $n$ points can be isometrically embedded in \h{a $(n-1)$-dimensional Euclidean space ($\ell_2^{n-1}$)}.
In a  later work, Blumenthal \cite{blumenthal1953} shows that any space which is isometrically embeddable in a Hilbert space has the $n$-point property. This single result applies to many metrics, including Euclidean, Cosine, Jensen-Shannon and Triangular \cite{Connor2016:HilbertExclusion}, and is sufficient for our purposes here.
\subsubsection{Dimensionality Reduction}
  aims to produce low-dimensional encodings of high-dimensional data,  preserving the local structure of some input data. 
  See \cite{Fodor2002survey,yang2006survey} for comprehensive surveys on this topic.

The \emph{Principal Component Analysis} (PCA) \cite{Jolliffe2014:PCA} is  the most popular of the techniques for unsupervised dimensionality reduction. The  idea is to find a linear transformation of $n$-dimensional   to $k$-dimensional  vectors ($k\leq n$) that best preserves the \textit{variance} of the input  data. 
Specifically, PCA projects the data along the direction of its first $k$ principal components, which are the eigenvectors of the covariance matrix of the (centered) input data. 

According to the \emph{Johnson-Lindenstrauss Flattening Lemma} (JL)  (see e.g. \cite[pag. 358]{matousek2013book}), a random projection can also be used to embed a finite set of $n$ euclidean vectors into a $k$-dimensional euclidean space space ($k<n$) with a ``small'' distortion. Specifically the Lemma asserts that for any $n$-points  of $\ell_2$ and every $0<\epsilon<1$ there is a mapping into $\ell_2^k$ that preserves all the interpoint distances within factor $1+\epsilon$, where $k=O(\epsilon^{-2}\log n)$. 
The low dimensional embedding given by the Johnson Lindenstrauss lemma is particularly simple to implement.
\subsubsection{General metric spaces} do not allow either PCA or JL as they require inspection of the coordinate space. \h{Mao et al.} \cite{Mao2010} \h{pointed out that multidimetional-methods can be indirectly applied to metric space by using the \emph{pivot space model}. In that case each metric object is represented by its distance to a finite set of pivots.}  

In the general metric space context,
perhaps the best known technique is \emph{metric Multidimensional Scaling} (MDS) \cite{Cox2008}. MDS aims to  preserve  \emph{inter-point distances}
using  spectral analysis.
 However, when the number $m$ of data points is  large the classical MDS is too expensive in practice due to a requirement for $O(m^2)$ distance computations and spectral decomposition of a $m\times m$ matrix. 

The \emph{Landmark MDS} (LMDS) \cite{DeSilva2004:LMDS} is a fast approximation of MDS. LMDS uses a set of $k$ \emph{landmark} points  to compute $k\times m$ distances of the data points from the pivots. It applies classical MSD to these  points and  uses a distance-based triangulation procedure to project the remaining data points. 


\subsubsection{LAESA}\cite{Mico1994} is a more tractable  mechanism which has been used for metric filtering, rather than approximate search.  $n$ reference objects are selected. For each element of the data, the distances to these points are recorded in a table. At query time, the distances between the query and each reference point are calculated. The table can then be scanned row at a time, and each distance compared; if, for any reference object $p_i$ and data object $s_j$ the absolute difference  $|d(q,p_i) - d(s_j,p_i)| > t$, then from triangle inequality it is impossible for $s_j$ to be within distance $t$ of the query, and the object need not be paged into the main memory.

%
%

\section{The N-Simplex Apical Space}
\label{sec_nsimplex_outline}
In this section we give an informal outline of our new observations on supermetric spaces. They are based on the fact mentioned above that, for any $(n+1)$ objects in the original space, there exists  a simplex in $\ell_2^n$  whose edge lengths correspond to the distances measured in the original space.

In \cite{Connor2016:HilbertExclusion,Connor2016:supermetricSISAP,Connor2016:supermetricIS} we showed a less general result, that any semi-metric which is isometrically embeddable in a Hilbert Space has the four-point property: that is, given all of the distances measured among any four objects in the space, it is possible to construct a tetrahedron in   3D Euclidean Space 
with edge lengths corresponding to those distances.
In \cite{Connor2016:supermetricSISAP,Connor2016:supermetricIS} we showed an important lower-bound property based on this tetrahedral embedding; this is illustrated in Figure \ref{fig:planar_projection}, extended  here with a matching upper-bound.
\begin{figure}[tbp]
	\centering
		{	\includegraphics[trim=30mm 0mm 30mm 0mm,width=0.5\textwidth] {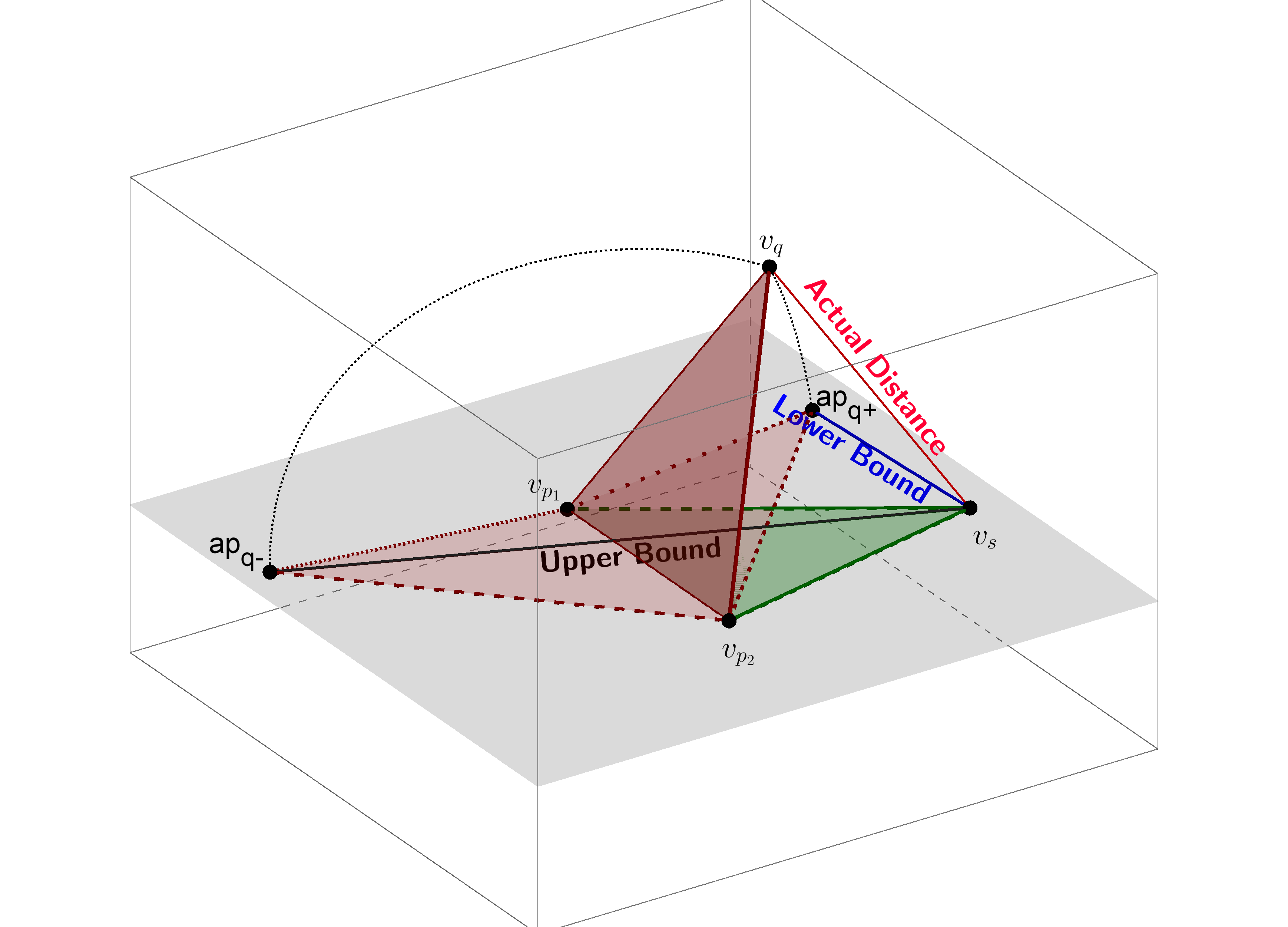}}
	\caption{Tetrahedral embedding of four points into 3D Euclidean space.}
	\label{fig:planar_projection}
\end{figure}

The case in point here is when four objects within the original space have been identified, but only five of the six possible distances have been measured. This corresponds to the situation of an indexing structure based on two reference objects, $p_1$ and $p_2$, which are chosen before a data set $S$ is organised  according to relative distances from these objects. The third object $s$ represents an arbitrary element of  $S$  which has been stored, and the fourth and final object $q$ represents a query over the data. For all possible $s$, we wish to identify those which may be within a threshold distance of $q$, based on some partition of the space constructed before $q$ was available.

Figure \ref{fig:planar_projection} shows an $\ell_2^3$ space into which these four objects have been projected, where for each element $a$ the notation $v_a$ is used to denote a corresponding point in the $\ell_2^3$ space. The only distance which has not been measured is $d(s,q)$; however the 4-point property means that the corresponding distance $\ell_2(v_s,v_q)$ must be able to form the final edge of a tetrahedron. From this Figure, the intuition of the upper and lower bounds on $d(s,q)$ is clear, through rotation of the triangle $v_{p_1}v_{p_2}v_q$ around the line $v_{p_1}v_{p_2}$ until it is coincident with the plane in which $v_{p_1}v_{p_2}v_s$ lies. The two possible orientations give the upper and lower bounds, corresponding to the distances between $v_s$ and the two apexes $ap_{q_-}$ and $ap_{q_+}$  of the two possible planar tetrahedra.

We now understand that this same intuition generalises into many dimensions. 
In the general form, we consider a set $p_i, i \in \{1 \dots n\}$, of $n$ reference objects, whose inter-object distances are used to form a  \emph{base} simplex $\sigma_0$, with vertices $v_{p_1},\dots, v_{p_n}$, in $(n-1)$ dimensions. This corresponds to the line segment $v_{p_1}v_{p_2}$ in the figure, this representing a two-vertex simplex in $\ell_2^1$. The simplex $\sigma_0$ is contained within a hyperplane of the  $\ell_2^n$ space, and the distances from object $s$ to each $p_i$ are used to calculate a new simplex $\sigma_s$, in $\ell_2^n$, consisting of a new apex point $v_s$ set above the base simplex $\sigma_0$. Note that there are two possible positions in $\ell_2^n$ for $v_s$, one on either side of the hyperplane containing $\sigma_0$; we denote these as $v_s^+$, and $v_s^-$ respectively. Now, given  the distances between object $q$ and all  $p_i$, we can  again construct two possible simplexes for $\sigma_q$ with two possible positions for $v_q$, which we denote by $v_q^+$ and $v_q^-$. 

Finally, we note that the act of rotating the triangle around its base also generalises to the concept of rotating the apex point of any simplex around the hyperplane containing its base simplex. Furthermore, the $n$-point property guarantees the existence of a simplex $\sigma_1$ in $\ell_2^{n+1}$ which preserves the distance $d(s,q)$ as $\ell_2(v_s,v_q)$. From  these observations we immediately have the following inequalities:
\begin{align*}
\ell_2^n(v_s^+, v_q^+) \quad \le \quad d(s,q) \quad \le  \quad \ell_2^n(v_s^+, v_q^-)
\end{align*}

To back up this intuition, we include proofs of these inequalities in the  appendix. Meanwhile, we answer the more pragmatic questions which allow these lower and upper bound properties to be useful in the context of similarity search.

\section{Constructing Simplexes from Edge Lengths}
In this section, we show an algorithm for determining Cartesian coordinates for the vertices of a  simplex, given only the distances between points. The algorithm is  inductive, at each stage allowing the apex of an $n$-dimensional simplex to be determined given the coordinates of an $(n-1)$-dimensional simplex, and the distances from the new apex to each vertex in the existing simplex. This is important because, given a fixed base simplex over which many new apexes are to be constructed, the time required to compute each one is linear with the number of dimensions.

A simplex is a generalisation of a triangle or a tetrahedron in arbitrary dimensions. 
%
In one dimension, the simplex is a line segment. In two dimensions it is a convex hull of a triangle, while in three dimensions it is the convex hull of a tetrahedron.
In general, the n-simplex of vertices $p_1,\dots,p_{n+1}$ equals the union of all the line segments joining $p_{n+1}$ to points of the $(n-1)$-simplex of vertices $p_1,\dots,p_{n}$. 

The structure of a simplex in $n$-dimensional space is given as an $n+1$ by $n$ matrix representing the cartesian coordinates of each vertex.
For example, the following matrix represents  four coordinates which are the vertices of a tetrahedron in 3D space:
\[
\begin{bmatrix}
0		&	0		&	0		\\
v_{2,1}	&	0		&	0		\\
v_{3,1}	&	v_{3,2}	&	0	\\
v_{4,1}	&	v_{4,2}	&	v_{4,3}
\end{bmatrix}
\]

For all such matrices $\Sigma$, the invariant that $v_{i,j} = 0$ whenever $j \ge i$ can be  maintained without loss of generality; for any simplex, this can be achieved by rotation and translation within the Euclidean space while maintaining the distances among all the vertices. Furthermore, if we restrict $v_{i,j} \ge 0$ whenever $j = i-1$ then in each row this component represents the \emph{altitude} of the $i^{th}$ point with respect to a base face represented by the matrix cut down from $\Sigma$ by selecting elements above and to the left of that entry.




\subsection{Simplex Construction}
	\begin{algorithm}[tbp]   
	{\KwIn{$n+1$ points $p_1,\dots,p_{n+1}\in (U,d)$}
		\KwOut{$n$-dimensional simplex in $\ell_2^n$ represented by the matrix $\Sigma\in \R{(n+1)\times n}$  }
		\BlankLine
		$\Sigma=  0\in \R{(n+1)\times n}$\;
		\If{$n=1$}{
			$\delta=d(p_1,p_2)$\;
			$\Sigma=\begin{bmatrix}
			0\\
			\delta
			\end{bmatrix}
			$\;
			\Return$\Sigma$;
		}
		$\Sigma_{Base}=$ nSimplexBuild($p_1,\dots,p_{n}$)\;
		$Distances=  0\in \R{n}$\;
		\textbf{for} $1\leq i\leq n$ set $Distances[i]=d(p_i,p_{n+1})$\;
		$newApex=$ ApexAddition($\Sigma_{Base},Distances$)\;
		\textbf{for} $1\leq i\leq n$ and  $1\leq j\leq i-1$ set $\Sigma[i][j]$ to $\Sigma_{Base}[i][j]$\;
		\textbf{for} $1\leq j\leq n$ set $\Sigma[n+1][j]$ to $newApex[j]$\;
		\Return $\Sigma$;
		\caption{nSimplexBuild}\label{alg:nsimplex}}
\end{algorithm}
\begin{algorithm}[tbp]
	\KwIn{A $(n-1)$-dimensional base simplex and the distances between a new (unknown) apex point and the vertices of the base simplex:
		\begin{align*}
		&\Sigma_{\text{Base}} = 
		\begin{bmatrix}
		0			&				&			&			\\
		v_{2,1}	&	0			&			\multicolumn{2}{c}{\text{\huge0}}		\\
		v_{3,1}	& v_{3,2}		&	\ddots		&				 \\
		\colon	& 				&\ddots		&  0			 	\\
		v_{n,1}	&				&\cdots		& v_{n,n-1}
		\end{bmatrix}\in \R{n\times n-1}
		  \\
		  &\\
		& Distances =
		\begin{bmatrix}
		\delta_1 & \cdots & \delta_n
		\end{bmatrix}\in \R{n}
		\end{align*}
	}
	\KwOut{The cartesian coordinates of the new apex point
	}
	\BlankLine
	$Output =
	\begin{bmatrix}
	\delta_1& 0 &\cdots & 0 
	\end{bmatrix} \in \R{n}
	$\; 
	\For{$i=2$ \KwTo $n$}{
		$l = \ell_2(\Sigma_{Base}[i],Output)$\;
		$\delta = Distances[i]$\;
		$x = \Sigma_{Base}[i][i - 1]$\;
		$y = Output[ i - 1]$\;
		$Output[i-1] = y - (\delta^2 - l^2)/2x$\;
		$Output[i] = +\sqrt{y^2 - (Output[i-1])^2}$\;
	}
	\Return $Output$
	\caption{ApexAddition}\label{alg:apex}
\end{algorithm}
This section gives an inductive
algorithm (Algorithm \ref{alg:nsimplex}) to construct a simplex in $n$ dimensions based only on the distances measured among $n+1$ points.

For the base case of a one-dimensional simplex (i.e. two points with a single distance $\delta$) the construction \h{is simply}
$\Sigma=\begin{bmatrix}
0\\
\delta
\end{bmatrix}
$\h{.}
 For an $n$-dimensional simplex, where $n \ge 2$, the distances among $n+1$ points are given. In this case, an $(n-1)$-dimensional simplex is first constructed using the first $n$ points. This simplex is used as a simplex base to which a new apex, the ${(n+1)}^{th}$ point, is added by the following \emph{ApexAddition} algorithm (Algorithm \ref{alg:apex}).
 
  For an arbitrary set of objects $s_i \in S$, the apex $\phi_{n}(s_i)$ can be pre-calculated. 
 When a query is performed,  only $n$ distances in the metric space require to be calculated to discover the new apex $\phi_{n}(q)$  in $\ell_2^n$.  
 
 In essence, the \emph{ApexAddition} algorithm is derived from exactly the same intuition as the lower-bound property explained earlier. Proofs of correctness for both the construction and the lower-bound property are included as an Appendix for the interested reader.

\subsection{Bounds}
Because of the method we use to build simplexes, the final coordinate always represents the altitude of the apex above the hyperplane containing the base simplex. Given this, two apexes exist, according to whether a positive or negative real number is inserted at the final step of the algorithm.

As a direct result of this observation, and those given in Section \ref{sec_nsimplex_outline}, we have the following bounds for any two objects $s_1$ and $s_2$ in the original space:

Let
\begin{align*}
\phi_{n}(s_1)  \quad &=\quad (x_1, x_2,\dots, x_{n-1},x_n)		\\
\phi_{n}(s_2)  \quad &=\quad (y_1, y_2,\dots, y_{n-1},y_n)
\end{align*}
then
\[
\sqrt{\sum_{i=1}^n(x_i - y_i)^2} \quad\le\quad d(s_1,s_2)\quad\le\quad
\sqrt{\sum_{i=1}^{n-1}(x_i - y_i)^2 + (x_n + y_n)^2} 
\]
From the structure of these calculations, it is apparent that they are  likely to converge rapidly around the true distance as the number of dimensions used becomes higher, as we will show in Section \ref{sec_distortion}. It can also be seen  that the cost of calculating both of these values together, especially in higher dimensions, is essentially the same as a simple $\ell_2$ calculation.

Finally, we note that the lower-bound function is a proper metric, but the upper-bound function is not even a semi-metric: even although it is a Euclidean distance in the apex space, one of the domain points is  constructed by reflection across a hyperplane and thus the distance between a pair of identical points is in general non-zero.

\section{Measuring Distortion}
\label{sec_distortion}
\begin{figure}[tbp]
	\centering
	\includegraphics[trim=20mm 80mm 25mm 85mm,width=0.475\textwidth]{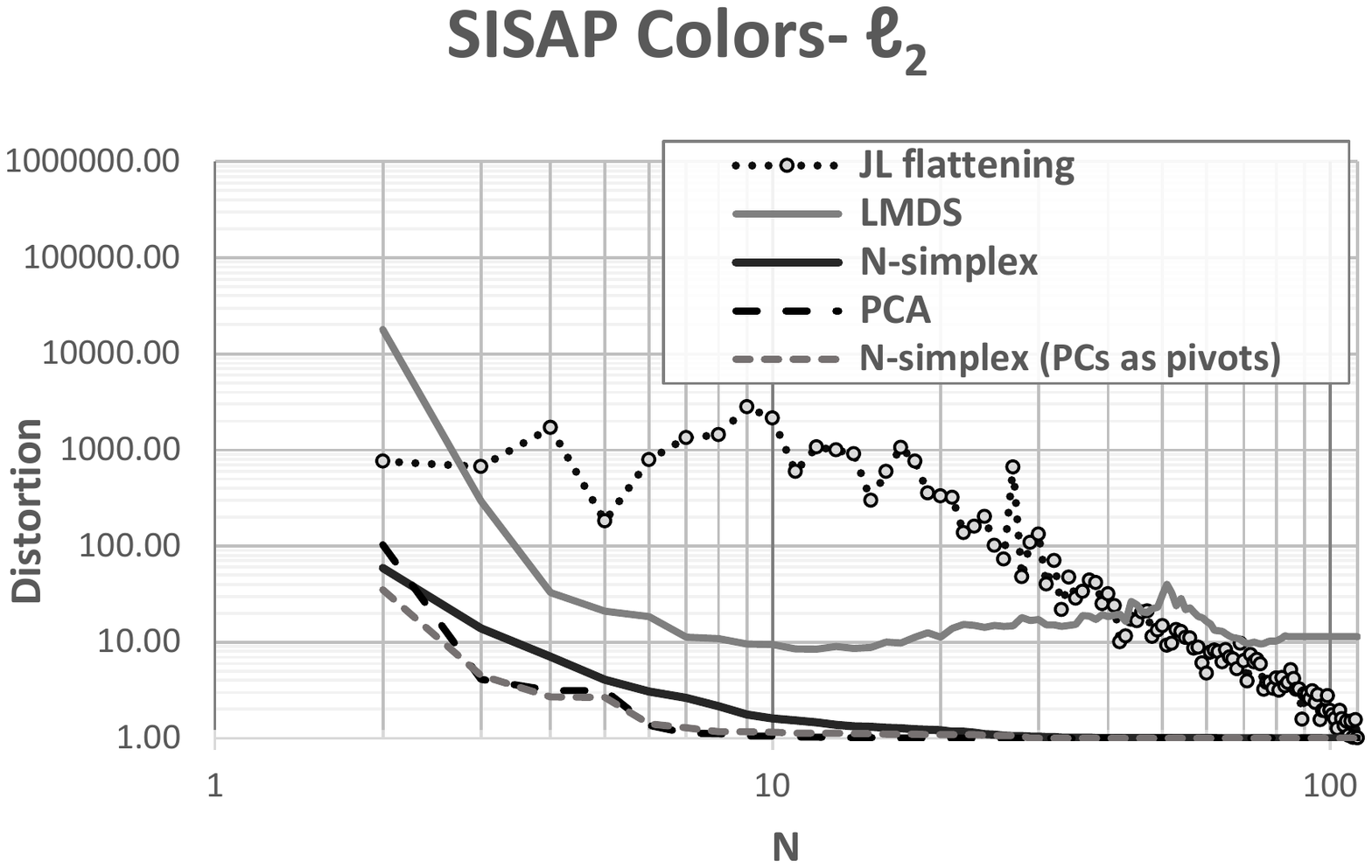}
	\quad  %
	\includegraphics[trim=20mm 80mm 25mm 85mm,width=0.475\textwidth]{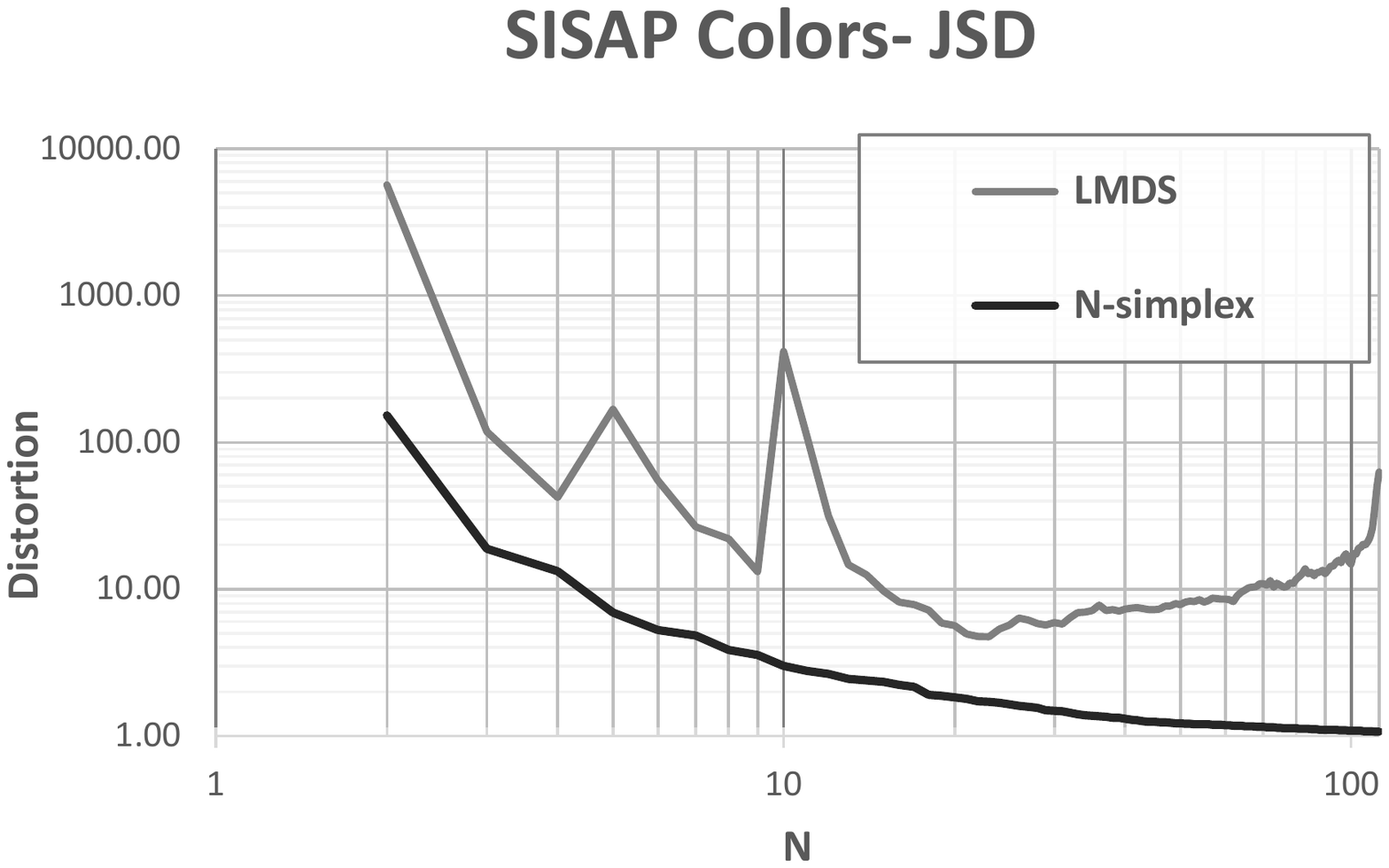}
\caption{Distortion measurements for various dimensionality reduction strategies for the \emph{colors} data set. The left figure gives measurements for Euclidean distance, the right for Jensen-Shannon distance where only LMDS and \ns are applicable. The \emph{colors} data set has 112 physical dimensions.}
	\label{fig:distortions}
\end{figure}

We define distortion for an approximation $(U',d')$ of a space $(U,d)$ mapped by a function $f:U \rightarrow U'$ as as the smallest $D$ such that, for some scaling factor $r$
\[
r \cdot d'(f(u_i),f(u_j)) \quad \le \quad d(u_i,u_j) \quad \le \quad D \cdot r \cdot d'(f(u_i),f(u_j))
\]
We have measured this for a number of different spaces, and present  results over the SISAP \emph{colors} benchmark set which are typical and easily reproducible. Summary results are shown in Figure \ref{fig:distortions}.

In each case, the X-axis represents the number of dimensions used for the representation, with the distortion plotted against this. For Euclidean distance, there are two entries for \ns: one for randomly-selected reference points, and the other where the choice of reference points is guided by the use of PCA. \h{In the latter case we select the first $n$ principal components (eigenvectors of the covariance matrix) as pivots.}

It can  be seen that \ns outperforms all other strategies except for PCA, which is not applicable to non-Euclidean spaces. LMDS is the only other mechanism applicable to general metric spaces%
\footnote{The authors note it works better for some metrics \h{than} others; in our understanding, it will work well only for spaces with the $n$-point property.}%
; this is a little more expensive than \ns to evaluate,
and performs relatively badly. The comparison with JL is a slightly unfair, as the JL  lemma applies only for very high dimensions in an evenly distributed space; we have tested such spaces, and JL is still out-performed at by \ns, especially at lower dimensions.

The distortion we show here is only for the lower-bound function of \ns. We have measured the upper-bound function also, which gives similar results. Unlike the lower-bound, the upper-bound is not a proper metric; however for non-metric approximate search it should be noted that the mean of the lower- and upper-bound functions give around half the distortion plotted here.

The implications of these results for exact search  should be noted. For Euclidean search, it seems that only around 20 dimensions will be required to perform a very accurate search, i.e. one-fifth of the original space. For Jensen-Shannon, more dimensions will be required, but the cost of the $\ell_2$ metric required to search the compressed space is around one-hundredth the cost of the original metric. In the next section we present experimental exact search results consistent with these observations.

\section{Exact Search: Indexing with n-Simplex}
\label{sec_exact_search}
In this section we examine the use of  \ns in the context of exact search, using the lower and upper-bound properties.
%
Any such mechanism can be viewed as similar to LAESA \cite{Mico1994}, in that there exists an underlying data structure
 which is a table of numbers, $n$ per original object, with the intention of using this table to exclude candidates which cannot be within a given search threshold. 
 
 In both cases, $n$ reference objects are chosen from the space. For LAESA, each row of the table is  filled, for one element of the data, with the distances from the candidate to each reference object. For \ns, each row is filled for one element of the data with the Cartesian coordinates of the new apex formed in $n$ dimensions by applying these distances to an $(n-1)$-dimensional simplex formed from the reference objects.

The table having been established, a query notionally proceeds by  measuring the distances from the query object to each reference point object.  In the case of LAESA, the metric for comparison is Chebyshev: that is, if any pairwise difference is greater than the query threshold, the object from which that row was derived cannot be a solution to the query. For \ns, the metric used is $\ell_2$: that is, if the apex represented in a row is further than the query threshold from the apex generated from the query, again the object from which that apex was derived cannot be a solution to the query.

In both cases, there are two ways of approaching the table search. It can be performed sequentially over the whole table, in which case either metric can be terminated within a row if the threshold is exceeded, without continuing to the end of the row. Alternatively the table can itself be re-indexed using a tree search structure: this can be implemented with only a few extra words per item by storing references into the table within the tree structure. Although this compromises the amount of space available for the table itself, it may avoid many of the individual row comparisons.

In the context of re-indexing we also note that, in the case of \ns, the Euclidean metric used over the table rows itself has the four-point property, and so the Hilbert Exclusion property as described in \cite{Connor2016:HilbertExclusion} may be used.

In all cases the result is a filtered set of candidate objects which is guaranteed to contain the correct solution set. In general, this set must be re-checked against the original metric, in the original space. For \ns however the  upper-bound condition is checked first; if this is less than the query threshold, then the object is guaranteed to be an element of the result set and does not require to be re-checked within the original space.

\subsection{Experiment - SISAP \emph{colors}}

\begin{table}[tb]
\caption{Elapsed Times - SISAP \emph{colors}, Euclidean distance.\\
All times are in seconds, for executing 11268 queries over 101414 data. The $Tree$ times are independent of the row and the mean is presented for simplicity.}
\centering
\begin{small}
\begin{tabular}{|c||c|c|c|c|c||c|c|c|c|c||c|c|c|c|c|}
\hline
&\multicolumn{5}{|c||}{$t_0=0.051768$}&\multicolumn{5}{|c||}{$t_1=0.082514$}&\multicolumn{5}{c|}{$t_2=0.131163$}\\
\hline
Dims&$L_{seq}$&$L_{rei}$&$N_{seq}$&$N_{rei}$&\textit{Tree}&$L_{seq}$&$L_{rei}$&$N_{seq}$&$N_{rei}$&\textit{Tree}&$L_{seq}$&$L_{rei}$&$N_{seq}$&$N_{rei}$&\textit{Tree}	\\
\hline
5&18.6&28.0&13.8&5.8&5.5&33.4&80.9&22.4&29.0&18.1&56.2&201.6&34.9&70.4&54.4\\
10&17.7&22.1&15.0&3.3&5.5&30.3&67.9&20.3&14.7&18.1&58.1&220.3&25.5&50.6&54.4\\
15&16.3&15.2&14.6&$\mathbf{3.0}$&5.5&26.7&59.7&20.2&12.1&18.1&45.8&159.5&$\mathbf{24.4}$&44.7&54.4\\
20&19.0&16.3&18.9&3.3&5.5&28.2&56.6&19.4&$\mathbf{11.5}$&18.1&46.8&189.3&27.8&48.3&54.4\\
25&22.5&16.9&20.4&3.4&5.5&27.4&56.8&22.3&13.4&18.1&45.5&167.5&26.2&40.1&54.4\\
30&20.9&16.8&20.4&3.5&5.5&28.6&57.3&24.5&13.6&18.1&45.9&181.2&28.5&45.1&54.4\\
35&22.0&16.4&21.3&3.9&5.5&28.7&65.0&22.5&13.9&18.1&43.9&163.0&31.2&44.9&54.4\\
40&23.1&17.3&22.1&4.0&5.5&28.8&55.9&22.8&14.3&18.1&49.4&180.5&34.2&46.1&54.4\\
45&22.5&18.7&22.2&4.4&5.5&32.0&61.5&27.7&15.0&18.1&48.5&169.8&37.1&44.9&54.4\\
50&21.3&17.1&18.9&4.5&5.5&32.0&59.0&24.0&15.5&18.1&55.2&207.6&34.5&45.3&54.4\\
\hline
\end{tabular}
\end{small}
\label{tab_euc_colors_times}
\end{table}%

We first apply these techniques to the SISAP \emph{colors} \cite{SISAP_man} data set, using three different supermetrics: Euclidean, Cosine, and Jensen-Shannon\footnote{For precise definitions of the non-Euclidean metrics used, see \cite{Connor2016:HilbertExclusion}.}.
We chose this data set because (a) it has only positive values and is therefore indexable by all of the metrics, and (b) it shows an interesting non-uniformity, in that its intrinsic dimensionality for all metrics is much less than its physical dimensionality (112). It should thus give an interesting ``real world" context to assess the relative value of the different mechanisms. For Euclidean distance, we used the  three benchmark thresholds; for the other metrics, we chose thresholds that return around 0.01\% of the data. In all cases the first 10\% of the file is used to query the remaining 90\%. \h{Pivots are randomly-selected both for LAESA and $n$-simplex approach.}

For each metric, we tested different mechanisms with different allocations of space: 5 to 50 numbers per data element, thus the  space used per object is between 4.5\% and 45\%  of the original. All results reported are for exact search, that is the initial filtering is followed by re-testing within the original space where required. Five different mechanism were tested, as follows:

\begin{description}
\item [sequential LAESA ($L_{seq}$)] each row of the table is scanned sequentially,  each element of each row is tested against the query and that row is abandoned if the absolute difference is greater than the threshold.
\item [reindexed LAESA ($L_{rei}$)] the data in the table is indexed using a monotone hyperplane tree, searched using the Chebyshev metric.
\item [sequential \ns ($N_{seq}$)]  each row of the table is scanned sequentially,  for each element of each row the square of the absolute difference is added to an accumulator, the row is abandoned if the accumulator exceeds the square of the threshold, and the upper-bound is applied if the end of the row is reached before re-checking in the original space.
\item [reindexed \ns ($N_{rei}$)] the data in the table is indexed using a monotone hyperplane tree using the Hilbert Exclusion property, and searched using the Euclidean metric; the upper-bound is applied for all results, before re-checking in the original space.
\item [normal indexing (${Tree}$)] the space is indexed using a monotone hyperplane tree with the Hilbert Exclusion property, without the use of reference points.
\end{description}

The monotone hyperplane tree is used as, in previous work, this has been found to be the best-performing simple indexing mechanism for use with Hilbert Exclusion.
\vspace{-10pt}
\subsubsection{Measurements} 
Three different figures are measured for each mechanism: the elapsed time, the number of original-space distance calculations performed and, in the case of the re-indexing mechanisms, the number of re-indexed space calculations. 
All code is available online for independent testing%
\footnote{https://richardconnor@bitbucket.org/richardconnor/metric-space-framework.git}.

The tests were run on a 2.8 GHz Intel Core i7, running on an otherwise bare machine without network interference. The code is written in Java, and all data sets used fit easily into the Java heap without paging or garbage collection occurring.
\vspace{-10pt}
\subsubsection{Results} 
	\begin{table}[tbp]
		\caption{Elapsed Times - SISAP \emph{colors} with Cosine and Jensen-Shannon distances, and a 30-dimensional generated Euclidean space. \\All times are in seconds. The generated Euclidean space is evenly distributed in $[0,1]^{30}$, and gives the elapsed time for executing 1,000 queries against 9,000 data, with a threshold calculated to return one result per million data ($t$=0.7269)}
			\centering
			\resizebox{\textwidth}{!}{
			\begin{tabular}{|c||c|c|c|c|c||c|c|c|c|c||c||c|c|c|c|c|}
				\hline
				&\multicolumn{10}{|c||}{SISAP \textit{colors}}&\multicolumn{6}{c|}{\multirow{2}{*}{30-dim $\ell_2^{30}$}}\\
				\cline{2-11}
				&\multicolumn{5}{|c||}{Cosine (t=0.042)}&\multicolumn{5}{|c||}{Jensen-Shannon (t=0.135)}&\multicolumn{6}{c|}{}\\
				\hline
				{\scriptsize Dims}&$L_{seq}$&$L_{rei}$&$N_{seq}$&$N_{rei}$&\textit{Tree}&$L_{seq}$&$L_{rei}$&$N_{seq}$&$N_{rei}$&\textit{Tree}&{\scriptsize Dims}&$L_{seq}$&$L_{rei}$&$N_{seq}$&$N_{rei}$&\textit{Tree}	\\
				\hline
				5&10.3&4.5&8.8&1.0&3.1&248.4&335.5&61.9&65.5&124.8&3&0.5&2.5&0.5&1.6&1.4\\
				10&9.8&3.4&10.4&0.8&3.1&155.3&233.2&29.0&29.3&124.8&6&0.5&2.3&0.5&1.8&1.4\\
				15&12.7&2.4&11.7&$\mathbf{0.7}$&3.1&103.5&163.2&22.3&17.2&124.8&9&0.5&2.4&0.4&1.3&1.4\\
				20&16.5&2.8&16.7&$\mathbf{0.7}$&3.1&95.7&162.8&23.8&$\mathbf{14.7}$&124.8&12&0.5&2.6&0.3&1.2&1.4\\
				25&17.9&2.8&17.7&0.8&3.1&87.2&155.6&25.9&16.1&124.8&15&0.5&2.8&0.3&1.0&1.4\\
				30&18.1&2.6&17.4&0.9&3.1&67.7&130.4&27.0&16.5&124.8&18&0.6&3.4&0.3&1.0&1.4\\
				35&17.7&3.1&17.1&1.1&3.1&69.6&136.3&27.9&17.2&124.8&21&0.6&3.3&$\mathbf{0.2}$&1.1&1.4\\
				40&18.1&3.0&18.1&1.0&3.1&62.4&131.2&27.8&17.1&124.8&24&0.7&2.9&$\mathbf{0.2}$&1.1&1.4\\
				45&17.4&2.7&18.2&1.1&3.1&61.1&133.4&29.7&18.4&124.8&27&0.7&3.5&0.3&1.2&1.4\\
				50&17.6&3.5&17.3&1.4&3.1&58.3&130.4&30.6&18.6&124.8&30&0.7&3.5&0.3&1.4&1.4\\
				\hline
			\end{tabular}
		}
		\label{tab_non_euc_colors_times}
	\end{table}%
	\begin{table}[tbp]
		\renewcommand{\arraystretch}{1.05}
		\caption{ Distance Calculations Performed  in Original and Re-indexed Space\\(figures given are thousands of calculations per query)}
		\centering
		\begin{small}
			\begin{tabular}{|c||c|c|c||c|c||c|c|c||c|c|}
				\hline
				&\multicolumn{5}{|c||}{Euclidean (t=0.051768)}&\multicolumn{5}{|c|}{Jensen-Shannon (t=0.135)}\\
				\cline{2-11}
				&\multicolumn{3}{|c||}{Original Space}&\multicolumn{2}{|c||}{Re-indexed}
				&\multicolumn{3}{|c||}{Original Space}&\multicolumn{2}{|c|}{Re-indexed}\\
				\hline
				Dims&$L$&$N$&\textit{Tree}&$L_{rei}$&$N_{rei}$&$L$&$N$&\textit{Tree}&$L_{rei}$&$N_{rei}$	\\
				\hline
				5&2.75&0.38&1.48&5.28&1.76&12.77&2.29&5.97&18.40&6.91\\
				10&1.33&0.05&1.48&4.40&1.23&7.81&0.58&5.97&19.66&6.32\\
				15&0.57&0.04&1.48&3.24&1.13&4.62&0.16&5.97&15.46&4.99\\
				20&0.51&0.03&1.48&3.42&1.15&3.89&0.11&5.97&15.85&4.80\\
				25&0.43&0.04&1.48&3.15&1.18&3.65&0.09&5.97&14.88&4.87\\
				30&0.37&0.04&1.48&3.02&1.21&2.53&0.08&5.97&13.83&4.70\\
				35&0.34&0.04&1.48&2.85&1.31&2.59&0.08&5.97&13.56&4.86\\
				40&0.33&0.04&1.48&2.95&1.29&2.14&0.08&5.97&13.48&4.64\\
				45&0.31&0.05&1.48&2.82&1.32&1.95&0.08&5.97&13.74&4.89\\
				50&0.27&0.05&1.48&2.57&1.33&1.83&0.08&5.97&12.63&4.87\\
				\hline
			\end{tabular}
		\end{small}
		\label{tab_orig_dists}
	\end{table}%

As can be seen in Table \ref{tab_euc_colors_times}, $N_{rei}$ consistently and significantly outperforms the normal index structure at between 15 and 25 dimensions, depending on the query threshold. It is also interesting to see that, as the query threshold increases, and therefore scalability decreases, $N_{seq}$ takes over as the most efficient mechanism, again with a ``sweet spot" at 15 dimensions.

Table \ref{tab_non_euc_colors_times} shows the same experiment performed with Cosine and Jensen-Shannon distances. In these cases, the extra relative cost saving from the more expensive metrics is very clear, with relative speedups of 4.5 and 8.5 times respectively. 
In the Jensen-Shannon tests, the relatively very high cost of the metric evaluation to some extent masks the difference between 
$N_{seq}$ and $N_{rei}$, but we  note that the latter maintains scalability while the former does not.  Finally, in the essentially intractable Euclidean space, with a relatively much smaller search threshold, $N_{seq}$ takes over as the fastest mechanism.
\vspace{-10pt}
\subsubsection{Scalability} 
Table \ref{tab_orig_dists} shows the actual number of distance measurements made, for Euclidean and Jensen-Shannon searches of the \emph{colors} data. The number of calls required in both the original and re-indexed spaces are given. Note that original-space calls are the same for both table-checked and re-indexed mechanisms; the number of original-space calls include those to the reference points, from which the accuracy of the \ns mechanism even in small dimensions can be appreciated. By 50 dimensions almost perfect accuracy is achieved for Euclidean search 50 original-space calculations are made,  but in fact even at 10 dimensions almost every apex value can be  deterministically determined as either a member or otherwise of the solution set based on its upper and lower bounds. At 20 dimensions, only 10 elements of the 101414-element data set have bounds which straddle the query threshold. \h{This indeed reflects the results presented in Figure} \ref{fig:distortions} \h{where it is shown that for $n\geq20$ the n-simplex lower bound is practically equivalent to the Euclidean distance to search \textit{colors} data.}

Equally interesting is the number of re-indexed space calls. This gave us a considerable surprise, and is the subject of further investigation: for \ns, these are generally  less than for the original space, including for tests made which are not presented here. This seems to hold for all data  other than perfectly evenly-distributed (generated sets), for which the scalability is the same. The implication is that the re-indexed metric has better scalability properties than the original, although we would have expected indexing over the lower-bound function to be less, rather than more, scalable.


\section{Conclusions and Further Work}

We have used the \ns technique to give best-recorded benchmark performance for exact search over the SISAP \emph{colors} data set for some different metrics. It should however be noted that here we are only trying to demonstrate the potential value of the \ns bounds mechanism in a simple and reproduceable context; as noted it is likely to be most effective in cases where the data set does not fit into memory, and where the metric used is very expensive. We emphasise that in all of our tests the whole data fits in main memory, and a recheck into the original space is relatively cheap.  We believe the real power of this technique will emerge with huge data sets and more expensive metrics, and is yet to be experienced.

\begin{small}
\subsubsection*{Acknowledgements}
The work was partially funded by Smart News, ``Social sensing for breaking news", co-funded by the Tuscany region under the FAR-FAS 2014 program, CUP CIPE D58C15000270008.
\bibliographystyle{plain}
\bibliography{bib/connor_abr,bib/general_abr}
\end{small}

\newpage
\appendix

\section*{Appendix}
\begin{lemma}[Correctness of the ApexAddition algorithm]
	Let $\Sigma_{\text{Base}}\in \R{n\times n-1}$ representing a $(n-1)$-dimensional simplex of vertices $\Sigma_{\text{Base}}[i]\in \ell_2^{n-1}$, with $\Sigma_{\text{Base}}[i][j]=0$ for all $j\geq i$ and $\Sigma_{\text{Base}}[n][n-1]\geq0$. Let $v_i$ the corresponding vertices in $\ell_2^n$ (obtained from $\Sigma_{\text{Base}}[i]$ by adding a zero to the end of the vector) and let $\delta_i$ the distance between an unknown apex point and the vertex $v_i$.
	Let $o=\begin{bmatrix} o_1& \dots & o_n\end{bmatrix}$ the output  of the \emph{ApexAddition} Algorithm. Then $o$ is a feasible apex, i.e. it is a point in $\R{n}$ satisfying $\ell_2(o,v_i)=\delta_i$ for all $1\leq i\leq n$. The last component $o_n$ is non-negative and represents the \emph{altitude} of $o$ with respect to a base face  $\Sigma_{\text{Base}}$.
\end{lemma} 
\begin{proof}
	It is sufficient to prove that the output $o=\begin{bmatrix} o_1& \dots & o_n\end{bmatrix}$ of the Algorithm \ref{alg:apex} has distance $\delta_i$ from the vertex $v_i$, i.e. satisfies the following equations
	\begin{equation}\label{eq:apexSystem}
	\begin{cases}
	o_1^2+\dots +  o_n^2=\delta_{1}^2 & \qquad(\ref{eq:apexSystem}.1)\\
	\qquad \colon\\
	\sum_{j=1}^{i-1} (v_{i,j}- o_j)^2+ \sum_{j=i}^n  o_j^2=\delta_i^2 & \qquad(\ref{eq:apexSystem}.i)\\
	\qquad \colon\\
	\sum_{j=1}^{n-1} (v_{n,j}-o_j)^2+  o_n^2=\delta_n^2 &\qquad (\ref{eq:apexSystem}.n)\\
	\end{cases}
	\end{equation}
	
	Note that the $i$-th component of the output $o$ is updated only at the iteration $i$ and $i+1$ of the \emph{ApexAddition} Algorithm. 
	So, if we denote with $o^{(i)}$ the output at the end of iteration $i$ we have:
	\begin{align}
	& o^{(1)}=\begin{bmatrix} \delta_1& 0&\dots & 0\end{bmatrix} \label{eq:o1}\\
	& o_i= o^{(h)}_i, \quad o_n=o^{(n)}_n,\quad o^{(i)}_h=0&  1\leq i < h\leq n   \label{eq:p0} \\
	&o_{i-1}=o_{i-1}^{(i-1)}-\frac{\delta_{i}^2-\sum_{j=1}^{i-2}(v_{i,j}-o_{j})^2-(v_{i,i-1}-o_{i-1}^{(i-1)})^2}{2v_{i,i-1}} &  2\leq i \leq  n \label{eq:p2b}\\
	&( o_{i-1})^2=( o_{i-1}^{(i-1)})^2- (o_{i}^{(i)})^2 &  1\leq i \leq  n-1 \label{eq:p3b}
	\end{align}
	By combining  Eq. \eqref{eq:p0} and \eqref{eq:p3b} we obtain $\sum_{j=i}^{n} o_j^2= (o_i^{(i)})^2$ for all $1\leq i \leq  n-2$,
	and so Eq. (\ref{eq:apexSystem}.1) clearly holds (case $i=1$). Moreover, 
	it follows that $o$ satisfies Eq. (\ref{eq:apexSystem}.$i$) for all $i=2,\dots,n$:
	\begin{align*}
	\sum_{j=1}^{i-1} (v_{i ,j}-o_j)^2+ \sum_{j=i }^n o_j^2& 
	=v_{i ,i-1}^2 -2v_{i ,i-1}\,o_{i-1}+\sum_{j=1}^{i-2} (v_{i ,i-1}-o_j)^2+(o_{i-1}^{(i-1)})^2 \stackrel{\eqref{eq:p2b}}{=} \delta_{i }^2
	\end{align*}
	\qed
\end{proof}
\begin{lemma}[n-Simplex Distance Constraint]
	Let $(U,d)$ a space $(n+2)$-embeddable in $\ell_2^{n+1}$. Let $p_1,\dots,p_n \in U$ and, for any $m\leq n$, let $\sigma_{m}$ the $(m-1)$-dimensional simplex generated from $p_1,\dots,p_{m}$ by using the \emph{nSimplexBuild} Algorithm.
	For any $x\in U$, let $x^{(m)}\in \ell_2^{m}$ the apex point with distance $d(x,p_1), \dots,$ $ d(x,p_m)$ from the vertices of $\sigma_{m}$, computed using the \emph{ApexAddition} Algorithm. 
	Then for all $q,s\in U$, 
	\begin{enumerate}
		\item 
		$\ell_2^{m-1} (s^{(m-1)},q^{(m-1)}) \leq \ell_2^{m}(s^{(m)},q^{(m)}) \quad\quad\quad \quad \mbox{for}\quad 2\leq m \leq n \label{eq:simplexlw}$ \label{cond1}
		\medskip
		\item 
		$g (s^{(m-1)},q^{(m-1)}) \geq g(s^{(m)},q^{(m)})  \,\quad\quad\quad\quad\quad\quad \mbox{for}\quad 2\leq m \leq n \label{eq:simplexub}$ \label{condub}
		\medskip
		\item 
		$\ell_2^n(s^{(n)},q^{(n)}) \leq d(s,q) \leq g(s^{(n)},q^{(n)})$
	\end{enumerate}
	where, for any $k\in \mathbb{N}$, $g:\ell_2^{k}\to \ell_2^{k}$ is defined as $g(x,y)=\sqrt{\sum_{i=1}^{k-1} (x_i-y_i)^2+(x_k+y_k)^2}$.
\end{lemma} 
\begin{proof}
	By construction, for any $m\leq n$ we have
	\begin{align}
	& x_i^{(m)}= x_i^{(m-1)} & i=1,\dots, m-2 \label{eq:c1}\\ 
	& x_i^{(i)}\geq 0 & i=1,\dots, m \label{eq:c3}\\
	&(x_{m-1}^{(m)})^2+( x_{m}^{(m)})^2 ={(x_{m-1}^{(m-1)})^2} \label{eq:c2}
	\end{align}
	Condition \ref{cond1} directly follows from Eq. \eqref{eq:c1}-\eqref{eq:c2}:
	\begin{align*}
	\ell_2^{m} (s^{(m)},  q^{(m)})^2& 
	=\ell_2^{m-1} (s^{(m-1)},q^{(m-1)})^2- (s^{(m-1)}_{m-1}-q^{(m-1)}_{m-1})^2+ \sum_{i=m-1}^{m}(s^{(m)}_i-q^{(m)}_i)^2\\
	&=\ell_2^{m-1} (s^{(m-1)},q^{(m-1)})^2 + 2\Big[-s^{(m)}_{m-1}q^{(m)}_{m-1}-s^{(m)}_{m}q^{(m)}_{m}\\
	&\qquad \qquad\qquad\qquad\qquad\quad+\sqrt{(s_{m-1}^{(m)})^2+(s_{m}^{(m)})^2 } \sqrt{(q_{m-1}^{(m)})^2+(q_{m}^{(m)})^2 }\Big]\\
	&\geq \ell_2^{m-1} (s^{(m-1)},q^{(m-1)})^2
	\end{align*}
	where the last passage follows from the Cauchy–Schwarz inequality
	\footnote{Cauchy–Schwarz inequality in two dimension is: $(a_1b_1+a_2b_2)^2\leq (a_1^2+a_2^2)(b_1^2+b_2^2)$ $\forall a_1,b_1,a_2,b_2 \in \mathbb{R}$, which implies $$(a_1b_1+a_2b_2)\leq \sqrt{(a_1^2+a_2^2)}\sqrt{(b_1^2+b_2^2)} \quad \forall a_1,b_1,a_2,b_2 \in \mathbb{R}$$}.
	
	Similarly, Condition \ref{condub} also holds: 
	\begin{align*}
	g (s^{(m)}, q^{(m)})^2&=g(s^{(m-1)},q^{(m-1)})^2 + 2\Big[-s^{(m)}_{m-1}q^{(m)}_{m-1}+s^{(m)}_{m}q^{(m)}_{m}\\
	&\qquad \qquad\qquad\qquad\qquad\quad-\sqrt{(s_{m-1}^{(m)})^2+(s_{m}^{(m)})^2 } \sqrt{(q_{m-1}^{(m)})^2+(q_{m}^{(m)})^2 }\Big]\\
	&\leq g(s^{(m-1)},q^{(m-1)})^2.
	\end{align*}

	Now we prove that $\ell_2^n(s^{(n)},q^{(n)})$ and $g(s^{(n)},q^{(n)})$ are, respectively, a lower bound and an upper bound for the actual distance $d(s,q)$.
	The main idea is using the simplex  $\sigma_{n}$ spanned by $p_1,\dots, p_n$ as a base face to  build the simplex  $\sigma_{n+1}$ spanned by $p_1,\dots, p_n, s$ and then use the latter as base face to build the simplex $\sigma_{n+2}$ spanned by $p_1,\dots, p_n, s,q$. In this way, we have an isometric embedding of $p_1,\dots, p_n, s,q $ into $\ell_2^{n+1}$ that is the function that maps  $p_1,\dots, p_n, s,q$  into the vertices  of $\sigma_{n+2}$. 
	So, given the base simplex $\sigma_{n}$ (represented by the matrix $\Sigma_{n}$), and the apex $s^{(n)}, q^{(n)}\in \ell_2^n$ we have that the simplex  $\sigma_{n+2}$  is represented by
	\begin{equation}
	\Sigma_{n+2} = 
	\left[\begin{array}{ccc|cc}
	\multicolumn{3}{c|}{\multirow{4}{*}{\large{$\Sigma_{n}$}}}& \multicolumn{2}{c}{\multirow{4}{*}{\large{$0$}}}\\
	\multicolumn{3}{c|}{}\\
	\multicolumn{3}{c|}{}   \\
	\hline
	s^{(n)}_1		&	\cdots		&  s^{(n)}_{n-1}&   s^{(n)}_n & 0	\\
	q^{(n)}_1	&	\cdots		& q^{(n)}_{n-1}&	q^{(n+1)}_{n}	   	& q^{(n+1)}_{n+1}	
	\end{array}\right]\in \R{n+2\times n+1} 
	\end{equation}
	where, by construction, $(q^{(n+1)}_{n+1})^2={(q^{(n)}_{n})^2-(q^{(n+1)}_{n})^2}$, $s^{(n)}_n,q^{(n+1)}_{n+1} \geq 0$, and $d(q,s)$ equals the euclidean distance between the two last rows of $\Sigma_{n+2}$.
	
	It follows that
	\begin{align}
	d(q,s)^2
	&= \sum_{i=1}^{n-1}(s^{(n)}_i-q^{(n)}_i)^2+(s^{(n)}_n)^2+(q^{(n)}_n)^2-2s^{(n)}_nq^{(n+1)}_n; 
	\end{align}
	and, since $q^{(n)}_n\geq |q^{(n+1)}_n|$, we have
	$$
	d(q,s)^2=\ell_2^n(s^{(n)},q^{(n)})^2 +2s^{(n)}_n (q^{(n)}_n-q^{(n+1)}_n) \geq \ell_2^n(s^{(n)},q^{(n)})^2,
	$$
	and
	$$
	d(q,s)^2=g(s^{(n)},q^{(n)})^2 -2s^{(n)}_n (q^{(n)}_n+q^{(n+1)}_n)\leq g(s^{(n)},q^{(n)})^2
	$$
	\qed
\end{proof}
\end{document}